\documentclass[letterpaper]{article} 
\usepackage{aaai24}  
\usepackage{times}  
\usepackage{helvet}  
\usepackage{courier}  
\usepackage[hyphens]{url}  
\usepackage{graphicx} 
\usepackage{natbib}  
\usepackage{caption} 
\usepackage{algorithm}
\usepackage{algorithmic}

\usepackage{newfloat}
\usepackage{listings}
\DeclareCaptionStyle{ruled}{labelfont=normalfont,labelsep=colon,strut=off} 
\floatstyle{ruled}
\newfloat{listing}{tb}{lst}{}
\floatname{listing}{Listing}
\nocopyright

\author {
    Zachary Wojtowicz
}
\affiliations{
    Harvard Economics \& Harvard Business School \\
    zwojtowicz@fas.harvard.edu
}

\usepackage{bibentry}

\usepackage{epigraph}
\usepackage{amsmath,amsfonts,amsfonts,amssymb,amsthm}

\def\R{\mathcal{R}}
\def\1{\mathbf{1}}

\def\N{\mathbb{N}}

\def\F{\mathcal{F}}

\newtheorem{theorem}{Theorem}

\newtheorem{definition}{Definition}
\newtheorem*{theorem*}{Theorem}

\title{When and Why is Persuasion Hard? A Computational Complexity Result}

\begin{document}

\maketitle

\begin{abstract}

As generative foundation models improve, they also tend to become more persuasive, raising concerns that AI automation will enable governments, firms, and other actors to manipulate beliefs with unprecedented scale and effectiveness at virtually no cost. The full economic and social ramifications of this trend have been difficult to foresee, however, given that we currently lack a complete theoretical understanding of why persuasion is costly for human labor to produce in the first place. This paper places human and AI agents on a common conceptual footing by formalizing \emph{informational persuasion} as a mathematical decision problem and characterizing its computational complexity. A novel proof establishes that persuasive messages are challenging to discover (NP-Hard) but easy to adopt if supplied by others (NP). This asymmetry helps explain why people are susceptible to persuasion, even in contexts where all relevant information is publicly available. The result also illuminates why litigation, strategic communication, and other persuasion-oriented activities have historically been so human capital intensive, and it provides a new theoretical basis for studying how AI will impact various industries.\footnote{This version incorporates an erratum that corrects a gap in the original proof. Thanks to Alban Grastien for spotting the gap.}

\end{abstract}

\section{Introduction}

Advocates argue that artificial intelligence will simultaneously accelerate the pace of scientific discovery and make existing knowledge more accessible \cite[\emph{e.g.}, by translating medical information into people's native languages;][]{vieira2021understanding}. However, the generality of artificial intelligence raises dual-use concerns: similar to how a predictive model intended to avoid drug toxicity can be ``inverted'' to increase the potency of chemical weapons \citep{urbina2022dual}, systems that are capable of explanation and creativity will also generally be potent instruments of manipulation and disinformation. A comprehensive assessment of the technology requires that we anticipate both its potential epistemic costs and benefits, ideally maximizing the latter while minimizing the former.

Recent work has demonstrated that large language models are highly effective at persuading people across a variety of tasks and domains \cite{matz2024potential,durmus2024persuasion,burtell2023artificial,shin2023enhancing,ahn2021ai,karinshak2023working,carrasco2024large,breum2023persuasive}, raising concerns about their potential use for widespread misinformation, manipulation, and deception \citep{allen2024real,kreps2022all,zellers2019defending}. These concerns are especially acute for people whose identifiable characteristics---such as race, gender, or sexual identity---subject them to higher rates of algorithmic persuasion and bias \cite{bar2023algorithmic,speicher2018potential,hannak2014measuring,mikians2012detecting}.

Large language models and other innovations have significantly enhanced the persuasive capacities of machines in recent years, but humans have been refining their powers of persuasion since antiquity (Aristotle's \emph{Rhetoric} dates from the 4th century BCE). The large-scale manufacture of persuasion is, perhaps unsurprisingly, a well-developed fixture of modern economies. One frequently-cited estimate holds that ``one-quarter of GDP is persuasion'' \citep{mccloskey1995one}: law, advertising, politics, science, public relations, and many other professions revolve---in whole or in part---around changing other people's minds. The advent of targeted digital advertising has, moreover, expanded the specificity and scope of persuasion \cite{wu2017attention,zuboff2019age}.

The economics of persuasion raise three important questions: (1) what does persuasion accomplish to justify such enormous expenditures; (2) why has persuasion historically commanded so much brain power; and (3) how will AI automation reshape persuasion and, with it, the broader epistemic landscape of society?

A variety of authors have advanced theories of the \emph{benefits} of persuasion---\emph{i.e.}, what persuaders can achieve in various circumstances \cite[\emph{e.g.},][]{milgrom1986relying,grossman1980disclosure,crawford1982strategic,kamenica2011bayesian,schwartzstein2021using,aina2021tailored}. Relatively few, however, have focused on the \emph{cost} of generating persuasive messages. In practice, such costs play an important role in determining when, where, and how persuasion takes place. The purpose of this paper is to formally establish a key driver of these costs---namely, the computational resources (natural or artificial) required to generate persuasive messages.

Given a suitable formalization of the persuasion problem, computational complexity theory can be used to characterize its production function. This approach has previously been applied to explain a variety of economic and social phenomena, such as incomplete uptake of public information \citep{aragones2005fact}, slow convergence to Nash equilibrium \citep{daskalakis2009complexity}, narrow choice bracketing \citep{camara2022computationally}, and persistent market inefficiencies \citep{spear1989learning}. Closest to the present paper is that of \citeauthor{dughmi2016algorithmic} \citeyearpar{dughmi2016algorithmic}, who characterize the complexity of the model of persuasion studied by \citeauthor{kamenica2011bayesian} \citeyearpar{kamenica2011bayesian}.

\section{Informational Persuasion}

\vspace{1.25em}

\begin{center}
    \parbox{.4\textwidth}{``1.11 The world is determined by the facts, \emph{and by these being all the facts.}'' (emphasis added)}
\end{center}

\hspace{8em}---Ludwig Wittgenstein, 1922

\vspace{1.25em}

The present paper studies \emph{informational persuasion}: the selective disclosure of private information or the use of models (narratives, stories, \emph{etc.}) to ``frame'' public information with the intention of increasing a counter-party's belief in a focal claim. This formalization of persuasion is highly general in that it assumes little structure beyond the standard probability axioms and naturally captures a wide variety of applications. It can be cast as a formal decision problem as follows.

\begin{samepage}
    \begin{definition}[\textbf{Informational Persuasion}]
        Let a probability space $(\Omega,\mathcal{G},\pi)$, focal event $E \subseteq \Omega$, set of facts $\mathcal{F} = \{F_i \subseteq \Omega \, | \, i \in I\}$, and threshold of belief $p \in (0,1]$ be given. Is there a subset of the facts $\R \subseteq \F$ that induces the receiver to raise their belief in the focal event $E$ above the threshold, $\pi(E|\R) \geq p$?
    \end{definition}
\end{samepage}

Although some instances of data withholding are nefarious (\emph{e.g.}, a scientist omitting disconfirmatory evidence from their report), attention and other constraints mean that informational persuasion is a normalized feature of many interactions (\emph{e.g.}, ``please submit no more than \emph{three} letters of recommendation''). The use of models (narratives, stories, \emph{etc.}) to characterize publicly available information is a frequent practice among lawyers, politicians, academics, and other professional persuaders, as has been documented extensively elsewhere \citep{shiller2017narrative,spiegler2016bayesian,eliaz2020model,roos2021narratives,andre2023narratives,benabou2018narratives,graeber2024stories}. Informational persuasion is isomorphic to the following problem, which asks whether a model can be used to ``frame'' public information \citep[see, \emph{e.g.},][]{schwartzstein2021using}.

\begin{samepage}
    \begin{definition}[\textbf{Model Selection}]
        Let an informational persuasion problem be given. Define the space of models $M = 2^\mathcal{F}$, each of which induces a liklihood function over the focal event $E$ as $\pi(E|m) = \frac{\pi\left((\bigcap_{i \in m} \mathcal{F}_i) \,\cap\, E\right)}{\pi\left(\bigcap_{i \in m} \mathcal{F}_i\right)}$. Let a prior $\eta \in \Delta(M)$ be given. Is there a model $m \in M$ such that $\pi(m|E) \geq p$?
    \end{definition}
\end{samepage}

Our main result, Theorem~\ref{thm:lying}, shows that informational persuasion is NP-Complete. Problems in this complexity class share two principle features: proposed solutions can be quickly verified but they cannot, in general, be quickly discovered.\footnotemark\ Informational persuasion is therefore structurally similar to mathematical proof in the sense that solutions are categorically easier to check than to construct. Indeed, verifying proofs of a bounded length in an axiomatic system such as Zermelo-Frankel is likewise NP-Complete, which implies that a fast algorithm for informational persuasion, were it to exist, could be be used to quickly prove arbitrary theorems (and vice-versa).

\footnotetext{The idea that this class of problems cannot be solved quickly (\emph{i.e.}, in polynomial time) is the famous $P\neq NP$ conjecture. Although open, it is widely believed to be true among computer scientists \citep{gasarch2019guest} on account of how strange the implications of its converse would be, among other reasons \citep{aaronson2016p}.}

A central puzzles in the persuasion literature is why people are susceptible to persuasion in the first place, especially in contexts where all information is publicly available? More specifically, why do people accept externally supplied persuasion rather than coming up with their own models, narratives, and stories? 

Many existing theories do not address this question, but rather directly assume that people are persuadable. The proof of Theorem~\ref{thm:lying} provides some indication of why this might be the case. The essential insight is that, in unrestricted inferential contexts, each new fact casts every other fact in a new light, shifting the credence it lends to the focal event. Translating this idea into Kolmogorov's set-theoretic model of probability theory reveals that informational persuasion is, in the worst case, a highly non-convex optimization problem.

\section{Formal Result}

Let $(\Omega,\mathcal{G},\pi)$ be a probability space. We study communication between a sender $s$ and receiver $r$. The sender has access to a  collection of \emph{private facts} $\F = \{F_1,F_2,\dots,F_N\}$ for $N \in \N$, each of which is a subset of $\Omega$. The sender selectively reports a subset of the facts $\R \subseteq \F$ to the receiver. The receiver updates beliefs on the basis of the newly reported facts. 

There is a \emph{focal event} $E \subseteq \Omega$. The sender is concerned with maximizing the receiver's posterior belief in $E$. Let $\pi(\cdot|\R)$ denote the sender's expectation of the receiver's posterior conditional on a report $\R$.

\begin{theorem}\label{thm:lying}
    Informational persuasion is NP-Complete.
\end{theorem}

\begin{proof}
        
        We show that persuasion is both in NP and NP-Hard.

        \vspace{.5em}\noindent\textbf{NP:} A valid report $\R$ constitutes a certificate for the informational persuasion decision problem. To verify such a report, one evaluates
        \begin{equation}
            \pi(E|R) = \frac{\pi\Big(E \cap \big(\cap_{R \in \R} \, R \big) \Big)}{\pi\Big((\cap_{R \in \R} \, R\Big)} \geq p
        \end{equation}
        Set intersection, function lookup (to evaluate $\pi$), and arithmetic operations can all be implemented in linear time by a deterministic Turing machine. Hence, informational persuasion is in NP.
        
        \vspace{.5em}\noindent\textbf{NP-Hard:} We prove that informational persuasion is NP-Hard using a reduction of the exact cover problem, which is known to be NP-Complete \citep{garey1979computers}. The exact cover problem can be stated as follows: For a set $S$ and collection of subsets $A = \{A_1,\dots,A_n\}$ such that $A_i \subseteq S$, is there a collection $B \subseteq A$ such that $B_i \cap B_j = \varnothing$ for all $i \neq j$ and  $\bigcup_i B_i = S$? In other words, within a given cover, can we find an exact cover?

        Let a set $S = \{s_1,\dots,s_m\}$ and cover $A = \{A_1,\dots,A_n\}$ be given as above. Construct a state space $\Omega = \{0,1,\dots,n\}^m$, a $\sigma$-algebra $\mathcal{G} = 2^\Omega$, and prior $\pi$ having uniform support over $\Omega$. The interpretation of a state $\omega \in \Omega$ is that each coordinate indexes which, if any, set in $A$ covers the corresponding element of $S$. Further, define $E = \{\omega \in \Omega \,|\, \omega_j > 0 \text{ for all } j\in\{1,\dots,m\}\}$ to be the set of states that constitute complete assignments, then let $p = 1$ and $\mathcal{F}_i = \{\omega \in \Omega \,|\, \omega_j = i \text{ for every } s_j \in A_i \}$ for each $i \in \{1,\dots,n\}$. 

        Now, $(\Omega,\mathcal{G},\pi)$, $E$, $\mathcal{F}$, and $p$ define an informational persuasion problem. We claim that this decision problem has the same truth-value as the original exact cover problem. Note that, per our construction, the informational persuasion evaluates to TRUE if and only if there exists a subset $\R \subseteq \mathcal{F}$ such that $\frac{\pi(E \cap (\cap_{F \in \mathcal{R}} F))}{\pi(\cap_{F \in \mathcal{R}} F)} = 1$. Letting $I \subseteq \{1,\dots,n\}$ denote the set of indices included in $\R$, this condition is equivalent to asserting that both: (i) $\cap_{i \in I} \mathcal F_i \neq \varnothing$ (otherwise, the conditional probability would be undefined); and (ii) $\cap_{i \in I} \mathcal F_i \subseteq E$. By construction, each coordinate $\omega_j$ uniquely indexes an assignment of $s_j$ to $A$, so the existence of an $\omega \in \cap_{i \in I} \mathcal F_i$ states that each element of $S$ is assigned to \emph{at most} one set in $A$ by $\mathcal R$. Moreover, $E$ is defined as the set of complete assignments, so $\cap_{i \in I}\mathcal F_i \subseteq E$ states that each element of $S$ is assigned to \emph{at least} one set in $A$ by $\mathcal R$.
 
        Taken together, this is precisely the condition that every element of $S$ is covered by exactly one element of $B = \{A_i \in A| i \in I\}$. In other words, an index $I$ solves the informational persuasion problem if and only if it solves the corresponding exact cover problem. Hence, informational persuasion is in NP-Hard.

\end{proof}

\section{Discussion}

Rational theories of belief formation (such as those typically assumed in statistical decision theory and economics) assume that people instantaneously update their internal state of belief to reflect the sum total of information they possess.\footnotemark\ The ramifications of even readily available information are not always apparent, however: it takes time to notice patterns, generate good explanations, rule out bad explanations, separate signal from noise, and generally transform the ``raw material'' of information into the ``finished products'' of knowledge and understanding. 

\footnotetext{Attention is generally conceived of as a ``bottleneck'' interposed between an economic agent and the external world \citep{loewenstein2023economics}. Even models which incorporate attentional frictions typically assume that people update instantaneously once information makes it through the bottleneck, however.}

Characterizing the computational complexity of informational persuasion directly informs our understanding of this process by characterizing the production technology of informational persuasion. Theorem~\ref{thm:lying} shows that the computational resource cost of informational persuasion can, in the worst case, grow exponentially in the number of facts considered. This helps explain why industries centered around persuasion, such as those identified by \cite{mccloskey1995one}, are human-capital intensive. It also informs conversations about how artificial intelligence, social media, and other information technologies will impact our shared ``epistemic commons.''

Not all persuasion problems are equally difficult, however. There are clearly situations where it is easy to ``cherry-pick'' examples to support a particular conclusion---for example, if each ``fact'' is a Gaussian draw and the focal belief is a one-tailed hypothesis test concerning the distribution's mean. Although left to future work, the present approach frames informational persuasion directly in terms of the Kolmogorov probability axioms, and therefore naturally bridges persuasion with statistical theory. The simplicity of persuasion in the above example, for example, follows directly from statistical sufficiency.


\bibliography{references}

\end{document}